\documentclass[12pt]{article}

\usepackage[titletoc,title]{appendix}
\usepackage[latin1]{inputenc}
\usepackage{amsmath}
\usepackage{amssymb}
\usepackage{booktabs}
\usepackage{graphicx}
\usepackage[margin=1.25in]{geometry}
\usepackage[bottom]{footmisc}
\usepackage{indentfirst}
\usepackage{endnotes}
\usepackage{mathabx}
\usepackage{enumerate}
\usepackage{rotating}
\usepackage{multirow}
\usepackage{booktabs,calc}
\usepackage{threeparttable}
\usepackage{float}
\usepackage{url}
\usepackage{amsmath,amsfonts}
\usepackage{amsthm}
\usepackage[onehalfspacing]{setspace}

\newcommand{\suchthat}{\;\ifnum\currentgrouptype=16 \middle\fi|\;}

\newtheorem{prop}{Proposition}
\newtheorem{lemma}{Lemma}
\newtheorem{cor}{Corollary}

\newtheorem{fact}{Fact}

\newtheorem{defn}{Definition}

\newcommand{\continuation}{??}

\theoremstyle{remark}
\newtheorem{remark}{Remark}

\usepackage{chngcntr}
\usepackage{graphicx}
\usepackage{amsmath}
\usepackage{thmtools}
\usepackage{amssymb}
\usepackage{parskip}
\usepackage{sgame}
\usepackage{color}
\usepackage{tikz}
\usetikzlibrary{trees,calc}

\usepackage{xcolor}

\DeclareMathOperator*{\argmax}{\arg\!\max}

\newcommand{\E}{\mathbb{E}}

\newcommand{\eps}{\varepsilon}

\usepackage{natbib}

\usepackage{url,color,hyperref}
\definecolor{darkblue}{rgb}{0.0,0.0,0.35}
\hypersetup{colorlinks,breaklinks,
            linkcolor=darkblue,urlcolor=darkblue,
            anchorcolor=darkblue,citecolor=darkblue}

\begin{document}

\title{\large{Exogenous Consideration and Extended Random Utility} \thanks{I thank Nail Kashaev, Nirav Mehta, Rory McGee, Salvador Navarro, and seminar participants at the University of Western Ontario for helpful comments. This research was supported by the SSHRC Insight Development Grant Program.}}

\author{ \small{Roy Allen}  \\
    \small{Department of Economics} \\
    \small{University of Western Ontario} \\
    \small{rallen46@uwo.ca}
}
\date{\small{ \today }} 

\maketitle

\begin{abstract}
In a consideration set model, an individual maximizes utility among the considered alternatives. I relate a consideration set additive random utility model to classic discrete choice and the extended additive random utility model, in which utility can be $-\infty$ for infeasible alternatives. When observable utility shifters are bounded, all three models are observationally equivalent. Moreover, they have the same counterfactual bounds and welfare formulas for changes in utility shifters like price. For attention interventions, welfare cannot change in the full consideration model but is completely unbounded in the limited consideration model. The identified set for consideration set probabilities has a minimal width for any bounded support of shifters, but with unbounded support it is a point: identification ``towards'' infinity does not resemble identification ``at'' infinity.
\end{abstract}

\newpage

\section{Introduction}
To many economists, ``discrete choice'' is synonymous with logit or the general nonparametric additive random utility model (ARUM) studied in \cite{mcfadden1981econometric}. ARUM is widely used, but has many flaws. One criticism is that the representation assumes individuals consider all alternatives, which has motivated analysis of models that do not assume full consideration. Questions that arise include: When is it possible to falsify the hypothesis of full consideration? When can we identify the probability an alternative is considered? How does limited consideration alter counterfactual or welfare analysis?

This paper studies these questions for two models that ``update'' the classic additive random utility model to allow limited consideration. The first is the extended additive random utility model, in which unobservable utility shocks can equal $-\infty$ for some alternatives. Such alternatives are unavailable, not considered, or arbitrarily undesirable. The second is the additive random utility model with consideration sets. In this model, an individual first forms a consideration set, and then maximizes utility among alternatives in this set. Consideration set formation can arbitrarily depend on unobservable taste shocks, but does not depend on utility indices or covariates and so I term it exogenous.

Overall, this paper finds a divide between results that are robust across all ARUM family models, and results that sharply differ.  Classical analysis concerning empirical content (over bounded sets), welfare, and counterfactuals is robust to consideration sets. The robust results establish that the procedural concern that individuals do not consider all alternatives is not \textit{automatically} an empirical concern, at least for the models studied here. For an attention intervention, which is outside of classical analysis, an extreme contrast emerges: the analyst cannot meaningfully bound welfare increases for an attention intervention for the consideration set model. While the nonparametric consideration set model is more general and allows new questions to be tackled, it is not always capable of answering those questions. This is relevant because a major motivation of the recent literature on consideration sets is to measure welfare changes given attention interventions. More structure is necessary or richer observables, such as menu variation \citep{manzini2014stochastic}, covariates that only shift attention \citep{goeree2008limited}, or the extra assumption of independence between consideration and preference heterogeneity \citep{abaluck2021consumers,barseghyan2021discrete}. I now elaborate on the results.

For the empirically important case of bounded regressors, all ARUM-type models are observationally equivalent. Thus, any behavior that can be described by an ARUM consideration set model can be described by a full consideration ARUM model. This spills over to counterfactual analysis involving a change in utility indices (e.g. a price change): all models reach the same conclusions concerning bounds on choice probabilities. To an outside observer, it is ``as if'' each individual considers all alternatives. This does not mean ARUM is a ``good'' consideration set model. Instead, it formally establishes that the empirical limitations of the hypothesis of full consideration cannot be separated from general empirical failures of ARUM. When ARUM is falsified, it is not immediate whether this is due to endogenous consideration, classic endogeneity concerns, or other concerns entirely such as failures of individual optimization.

I next study identification of consideration set probabilities. I find a severe discontinuity. The identified set of a consideration probability has a minimal width for \textit{any} bounded set of utility index variation, yet with unbounded variation it is a point. Thus, large bounded support does not resemble full support. The intuition for this is we cannot distinguish $\eps_k$ being ``very negative'' from $\eps_k = -\infty$, interpreted as not considered. This discontinuity also manifests in counterfactual analysis with an attention intervention that makes an alternative always be considered.

Welfare analysis is subtle in consideration set models. \cite{masatlioglu2012revealed} highlight that the principle of revealed preference breaks. When an alternative is chosen, it is not obvious which alternatives it is preferred to. This makes direct utility comparisons challenging, but I show indirect utility comparisons are straightforward. Specifically, for a change in utility indices such as prices, the change in average indirect utility (or social surplus \citep{mcfadden1981econometric}) is identified. The reason for this is indirect utility calculations operate over a maximizing envelope and are unaffected by points below the envelope. Here, for all models, choice probabilities are the derivative of average indirect utility and we can integrate to identify changes in welfare. Thus, classic ``consumer surplus'' formulas are robust to exogenous limited consideration. Such formulas -- or local versions as in ``sufficient statistics'' \cite{chetty2009sufficient} -- are widely used empirically.

Next I analyze an attention intervention that makes a specific alternative always be considered. In the classic model there is no scope for a change. In contrast, for the limited consideration model, welfare bounds are the trivial bounds $[0,\infty)$. The model cannot rule out no welfare change, allows an arbitrarily high welfare change, and the bounds do not depend on choice probabilities. The reason for this is simple: individuals who do not consider an alternative may have an arbitrarily high latent utility shock for it. This makes analysis similar to the question of assessing the welfare impact of a new good, which requires enough structure. The benchmark model of exogenous consideration is useful because this result formally establishes that a model that generalizes this setup in one way -- such as by relaxing the exogenous consideration assumption -- must specialize in another to guarantee a finite welfare bound.

The primary analysis assumes utility indices for each alternative are known in order to focus on the role of unobservables. I show observational equivalence extends to the case of unknown utility indices, provided they are continuous and covariates vary over a bounded set. I show that in fact, the identified set for utility indices is the same for each model under these assumptions. This means that identification results for utility indices for one ARUM-type model can be ported to the others.\footnote{Results can always go from the more general consideration set models to classic ARUM, e.g. \cite{allen2019identification} covers a generalization of the extended ARUM. Porting an identification result from classic ARUM to the other models requires allowing choice probabilities to be separated from 1 to allow limited consideration. This rules out certain classic nonparametric identification results as in Theorem 4 in \cite{matzkin1993nonparametric}.}

The specific exogenous consideration set model studied here appears to be new, but I mention some setups related to exogenous consideration. The consideration set models of \cite{manzini2014stochastic}, \cite{aguiar2017random}, and \cite{jagabathula2021demand} are formally random utility models in the sense of \cite{block1960random}, when using menu variation. Consideration set models that can be written as random utility models can be interpreted as exogenous consideration set models: alternatives in the consideration set have a random utility that is well separated above the other alternatives, and the random utilities are independent of the menu and thus exogenous. In a setting of decisions under uncertainty, \cite{barseghyan2021discrete} assumes covariate shifters are independent of consideration to identify preferences and consideration sets. Likewise, \cite{abaluck2021consumers} has unobservables independent of prices that shift desirability. Though these models may be general random utility models due to the independence assumption, none of these models are additive random utility models.

This paper is part of a large literature studying choice sets that are latent. A classic paper is  \cite{manski1977structure}. A choice set could be latent for many reasons, including a good going out of stock \citep{conlon2013demand}. The term ``consideration set'' is not standardized in economics, but often refers to a latent mental choice set that is not modelled as an endogenous solution to an optimization problem (e.g. \cite{masatlioglu2012revealed}).\footnote{Outside of economics, more emphasis is placed on the formation of the consideration set itself, e.g. \cite{campbell1969existence}, \cite{narayana1975consumer}, \cite{hauser1990evaluation}, and \cite{roberts1991development}. See \cite{honka2019empirical} and \cite{crawford2021survey} for further references spanning economics, transportation, and marketing.} This distinguishes the terminology from explicit multi-stage choice in which a first stage is optimally chosen,\footnote{Exceptions include \cite{caplin2019rational} and parts of \cite{aguiar2023random}.} as in classic nested logit and certain models of search \citep{weitzman1978optimal}. The theoretical literature has studied empirical content and identification focusing on variation in the set of available alternatives, termed menu variation (\cite{manzini2014stochastic}, \cite{brady2016menu}, \cite{aguiar2017random}, \cite{kashaev2022random}). Econometric and empirical analysis has focused on covariate variation (\cite{chiang1998markov}, \cite{barseghyan2021discrete}, \cite{abaluck2021consumers}, \cite{lu2022estimating}, \cite{crawford2021survey}), including the use of exclusion restrictions that alter consideration but not preferences \citep{bacsar2004parameterized,goeree2008limited}. See Chapter 11 in \cite{strzalecki2023} for additional references and explicit links between menu and covariate variation.\footnote{See also \cite{fosgerau2013choice}, \cite{kawaguchi2021designing}, and \cite{allen2022revealed} in discrete choice, and \cite{he2023identification} and \cite{agarwal2022demand} in matching.} In contrast with most of this literature, the goal here is not to separately identify consideration and preferences. This paper instead shows that classic ARUM can be interpreted as a consideration set model. In this sense the paper is closer to \cite{matvejka2015rational} and \cite{fosgerau2020discrete}, who interpret classic discrete choice through the lens of rational inattention.

The rest of this paper is organized as follows. Section~\ref{sec:models} defines the models. Section~\ref{sec:obs} presents observational equivalence results, and then studies identification of consideration sets and the related question of distiguishing between the models. Section~\ref{sec:counter} studies counterfactual analysis with a change in covariates or attention. Section~\ref{sec:welfare} studies welfare analysis. Section~\ref{sec:unknownutility} generalizes the baseline setup to cover unknown utility indices. Section~\ref{sec:disc} presents concluding remarks.

\section{The Models} \label{sec:models}
This paper formalizes models in terms of structural choice probabilities $p : U \rightarrow \Delta^K$. The number of alternatives is $K \geq 2$, $\Delta^K$ is the probability simplex consisting of $K$-dimensional vectors that sum to $1$ and are non-negative. The argument of $p$ is a $K$-dimensional vector $u \in U$, interpreted as a vector of utility indices. Thus, $p(u)$ is the vector of choice probabilities given utility indices, $p_k(u)$ is the probability of choosing alternative $k$, and $u_k$ is the utility index for alternative $k$. I treat the utility indices as known in the core analysis. I study observational equivalence and identification questions using variation in utility indices.\footnote{An alternative interpretation of our core analysis is that $u_k$ is negative price, in which we study price variation. In Section~\ref{sec:unknownutility} we study unknown utility indices so that $u_k$ is replaced by $v_k(x_k)$ for a vector $x_k$ of observable covariates and unknown function $v_k$. This covers variation in covariates when the utility index is not known.} The set $U \subseteq \mathbb{R}^K$ is the set of utility indices for which $p(u)$ is identified. 

I first define the extended additive random utility model and classic additive random utility model. In these models, the individual chooses the alternative $k$ that maximizes $u_k + \eps_k$. Here, $u = (u_1, \ldots, u_K)$ is observable to each agent and the econometrician, and $\eps = (\eps_1, \ldots, \eps_K)$ is unobservable to the econometrician. I formalize consistency of choice probabilities with models below.
\begin{defn} \label{defarume}
Choice probabilities $p$ are consistent with the extended additive random utility model (ARUM-E) if there is a distribution $\mu$ over $\eps$ such that:
\begin{enumerate}[(i)]
    \item $\eps_k < \infty$ $\mu$-a.s. for each $k$.
    \item $p_k(u) = \Pr_{\mu}(u_k + \eps_k > \max_{j \neq k} u_j + \eps_j)$ for every $u \in U$ and $k$.
\end{enumerate}
If in addition $\mu$ can be chosen such that $\eps_k > -\infty$ $\mu$-a.s. for each $k$, then $p$ is consistent with the classic additive random utility model (ARUM).
\end{defn}
Any such distribution $\mu$ over $\eps$ is said to \textit{rationalize} the corresponding model. This paper rules out endogeneity between covariates and unobservables because it treats utility indices as parameters that change while the unobservables have the same distributions. Part (ii) assumes a unique maximizer with probability $1$, for any value of utility indices $u \in U$. It rules out utility ties to avoid discussing tie-breaking rules for each model.

Classic ARUM is the model studied in \cite{mcfadden1981econometric} among many others, and in econometrics is often just referred to as the random utility model. ARUM-E is a generalization of ARUM. The definition of ARUM-E as a stand-alone model appears to be new, but it is a special case of the perturbed utility model studied in \cite{allen2019identification}. When $\eps_k = -\infty$, alternative $k$ is either unavailable, not considered, or arbitrarily bad.\footnote{Page 4 in \cite{fosgerau2013choice} mentions allowing $u_k = -\infty$ to accommodate menu variation, i.e. $k$ is unavailable. This paper makes $u$ be finite and thus does not study menu variation.} Thus in ARUM-E, each individual chooses among alternatives with finite $\eps_k$ so as to maximize utility. Note that since a maximizer exists with probability $1$, some $\eps_k$ is finite with probability $1$.

I now formalize the additive random utility model with consideration sets. Let $\mathcal{S}$ denote the set of nonempty subsets of $\{1, \ldots, K\}$. I introduce $\eta \in H$ as an unobservable that alters consideration. The attention function $S : H \rightarrow \mathcal{S}$ describes which sets are considered given $\eta$. Each individual chooses to maximize $u_k + \eps_k$ among alternatives $k$ in the consideration set $S(\eta)$.
\begin{defn} \label{defcs}
Choice probabilities $p$ are consistent with the consideration set additive random utility model (ARUM-CS) if there is a distribution $\nu$ over $(\eps,\eta)$ such that:
\begin{enumerate}[(i)]
    \item $-\infty < \eps_k < \infty$ $\nu$-a.s. for each $k$.
    \item $p_k(u) = \Pr_{\nu}(\{ u_k + \eps_k > \max_{j \in S(\eta) : j \neq k} u_j + \eps_j \} \cap \{ k \in S(\eta)\})$ for every $u \in U$ and $k$.
\end{enumerate}
\end{defn}
Part (i) restricts each $\eps_k$ to be finite. Part (ii) models that $k$ is chosen when it maximizes utility among alternatives in the consideration set. Note that this definition allows variables determining consideration sets ($\eta$) to be arbitrarily correlated with the other shocks ($\eps$). 

To close the setup, we collect the definitions of distributions that rationalize $p$ according to the respective models.

\begin{defn} \label{defrat}
Let $\mathcal{M}^{ARUM}$ denote the set of distributions over real-valued $\eps$ that rationalize $p$ according to ARUM. Let $\mathcal{M}^{ARUM-E}$ denote the set of distributions over extended real-valued $\eps$ that rationalize $p$ according to ARUM-E. Let $\mathcal{M}^{ARUM-CS}$ denote the set of distributions over $(\eps,\eta)$ that rationalize $p$ according to ARUM-CS.
\end{defn}

\begin{remark}[Extended Reals]
ARUM-E requires working with addition in the extended reals. This paper defines $a + -\infty = -\infty$ for any $a \neq \infty$. Thus, when $\eps_k = -\infty$ we have $u_k + \eps_k = -\infty$ because the utility index $u_k$ is finite. Note that because $a > b$ only when $a > -\infty$, part (ii) of the definition of ARUM-E implies that for each $u \in U$, with probability $1$, $\eps_k$ is finite for some $k$.
\end{remark}

\begin{remark}[Normalizations]
For ARUM and ARUM-CS, it is a normalization to set $u_1 + \eps_1 = 0$ for alternative $1$, so that we can interpret the remaining $u_k + \eps_k$ numbers as differences relative to alternative $1$. By normalization I mean the empirical content of the model is the same with this extra assumption. The equality is not a normalization in ARUM-E because setting $u_1 + \eps_1 = 0$ implies that alternative $1$ is always considered.\footnote{This would imply that $\sup_{u \in U} p_1(u) = 1$ if $U = \mathbb{R}^K$ for example. I note that it is a normalization to set $u_1 = 0$ because utility here is always finite.} In particular, ARUM-E should not generally be interpreted as a model in differences. If $\eps_{1} = - \infty$ for example, then $\eps_k - \eps_1$ is not defined. Thus for ARUM-E, the analyst should \textit{not} make an assumption that one alternative has a set utility of $0$, unless the analyst also wishes to assume the alternative is always considered.
\end{remark}

\section{Observational Equivalence} \label{sec:obs}
I first present conditions under which the models are observationally equivalent, and then in Section~\ref{sec:iddisting} study how to distinguish between the models and the closely related question of identification of features of the distribution of consideration sets.

Say that $U$ has bounded utility differences if $\sup_{u \in U} |u_k - u_j| < \infty$, for any alternatives $j$ and $k$.
\begin{prop} \label{prop:obseq}
\begin{enumerate}[(i)]
    \item For any $U \subseteq \mathbb{R}^K$, ARUM-E and ARUM-CS are observationally equivalent, i.e. $p$ is consistent with ARUM-E if and only if it is consistent with ARUM-CS.
    \item If $U$ has bounded utility differences, then ARUM, ARUM-E, and ARUM-CS are observationally equivalent.
\end{enumerate}
\end{prop}

Part (i) states that we can never distinguish between ARUM-E and ARUM-CS.\footnote{See Proposition 5 in \cite{barseghyan2021discrete} for a related random utility model.} Working with the extended reals thus provides an alternative representation to a two-step choice model.\footnote{Encoding constraints in this manner is used in convex analysis, cf. \cite{rockafellar2015convex}.} The more noteworthy result is (ii), which shows that when utility differences are bounded, it is impossible to distinguish such models from the classic additive random utility model.\footnote{Proposition 2 in \cite{jagabathula2021demand} shows observational equivalence of a general random utility model (with nonseparable shocks) and a general consideration set model. We differ by considering different models and by using variation in latent utility indices and not variation in the set of available alternatives.} Thus, while it may be intuitively unappealing to assume individuals consider all alternatives, this concern can be addressed within ARUM by allowing a sufficiently flexible distribution of disturbances. Note that ARUM can be interpreted as a special case of ARUM-CS in which $\eta$ is constant and thus independent of $\eps$. Thus, when $U$ has bounded utility differences, adding independence between $\eta$ and $\eps$ is without loss of generality for studying empirical content.

Observational equivalence over bounded sets means that empirical content results that use local structure have equivalent translations between the models. Thus, local empirical content characterizations of ARUM in \cite{koning1994compatibility} and \cite{koning2003discrete} are also characterizations of ARUM-E and ARUM-CS for bounded sets.\footnote{Conditions in \cite{koning1994compatibility} are the classical Williams-Daly-Zachary-McFadden conditions proven to characterize ARUM in \citep{mcfadden1981econometric}, except dropping the requirement that probabilities limit to $0$ and $1$ at extreme values of covariates.} Or stated differently, their ARUM results are automatically robust to consideration sets.

I next study when it is possible to distinguish between ARUM and the other models, and the closely related question of identification of the distribution of latent feasibility sets.

\subsection{Identification and Distinguishability} \label{sec:iddisting}
I begin by presenting sharp identification results for certain consideration set probabilities. I then use these to establish that ARUM can be distinguished from ARUM-E (or ARUM-CS) if and only if certain marginal consideration set probabilities are identified. With Proposition~\ref{prop:obseq} this establishes that in general, unbounded utility differences are necessary for point identification of these marginal probabilities.

The identification results concern probabilities of the form $\Pr(\eps_k > -\infty)$ for ARUM-E, and $\Pr(k \in S(\eta))$ for ARUM-CS, i.e. the probability $k$ is considered. These probabilities are relevant because they must be 1 for ARUM. I define identified sets for these marginal probabilities, which reflect that there are typically multiple distributions that can rationalize the structural choice probabilities. I define the identified set for $\Pr(\eps_k > -\infty)$ as the union over all possible distributions that rationalize $p$ according to ARUM-E:
\[
\Theta^{k,E} = \cup_{\mu \in \mathcal{M}^{ARUM-E}} \Pr_{\mu} (\eps_k > -\infty).
\]
Similarly for $\Pr(k \in S(\eta))$ I define
\[
  \Theta^{k,CS} = \cup_{\nu \in \mathcal{M}^{ARUM-CS}} \Pr_{\nu}(k \in S(\eta)).
\]
\begin{lemma} \label{lem:idsetequiv}
The identified sets for marginal consideration probabilities are convex and satisfy
    \[
    \Theta^{k,E} = \Theta^{k,CS}.
    \]
    Moreover, $\sup_{u \in U} p_k(u)$ is a lower bound for each set. That is, for any $c \in \Theta^{k,E}$, it follows that $\sup_{u \in U} p_k(u) \leq c$.
\end{lemma}
Page 1980 in \cite{barseghyan2021discrete} notes a lower bound in a related context.

Distinguishability and identification questions hinge on whether $k$ can be made \textit{extremely attractive}, i.e. when
\[
\sup_{u \in U} \min_{j \neq k} \{ u_k - u_j \} = \infty.
\]
Let $\mathcal{K}^{EA}$ be the set of alternatives that can be made extremely attractive. This set $\mathcal{K}^{EA}$ is empty if and only if $U$ has bounded utility differences. As a final definition, say $u \in U$ is $k$-maximal if for any $w \in U$ and $j$, $u_k - u_j \geq w_k - w_j$.
\begin{prop} \label{prop:id}
Suppose $p$ is consistent with ARUM-E or equivalently ARUM-CS.
\begin{enumerate}[(i)]
    \item If $k \in \mathcal{K}^{EA}$, then
    \[
    \Theta^{k,E} = \Theta^{k,CS} = \sup_{u \in U} p_k(u).
    \]
    \item If $U$ has bounded utility differences (equivalently, $\mathcal{K}^{EA}$ is empty) and $\Theta^{k,E}$ is a singleton, then
    \[
    \Theta^{k,E} = \Theta^{k,CS} = \sup_{u \in U} p_k(u) = 1.
    \]
    \item If $U$ has bounded utility differences (equivalently, $\mathcal{K}^{EA}$ is empty) and contains a $k$-maximal point $u^*$, then
    \[
    \Theta^{k,E} = \Theta^{k,CS} = \left[p_k(u^*),1 \right] =  \left[ \sup_{u \in U} p_k(u), 1 \right].
    \]
\end{enumerate}
\end{prop}

Part (i) states that point identification of marginal consideration set probabilities is possible for any alternative that can be made extremely attractive. Part (ii) states that with bounded utility differences, point identification of these marginal consideration probabilities is equivalent to $\sup_{u \in U} p_k(u) = 1$. Part (iii) characterizes the identified set for marginal consideration probabilities for domains that contain a $k$-maximal point. An example of a set that contains $k$-maximal points for any $k$ is U a Cartesian product of compact intervals or more generally, a product of compact subsets of $\mathbb{R}$. One such set is $\{ 0, 1\}^K$, which has a $k$-maximal element for each $k$. For example, the $1$-maximal element is $(1,0, \ldots, 0)$. 

Unbounded regressors have been used previously to identify the distribution of consideration sets in a variety of settings \citep{abaluck2021consumers,barseghyan2021discrete,hyung2021,kashaev2023peer}.\footnote{\cite{abaluck2021consumers} identify changes in consideration set probabilities as prices change, for a different generalization of ARUM. Loosely, $S(\eta)$ can depend on $u$ (negative prices) in their setup in certain ways. That paper requires prices to limit to infinity to identify the levels of consideration set probabilities.} The primary novelty here is to highlight that for the ARUM family studied here, large support is typically necessary for point identification of the marginal distributions of consideration sets.

I present three implications of Proposition~\ref{prop:id}. To state the first result I generalize the notation a bit to accommodate a growing domain over which $p$ is identified. To that end, for any $U \subseteq \mathbb{R}^K$ let $\Theta^{k,E}(U) = \cup_{\mu \in \mathcal{M}^{ARUM-E}(U)} \Pr_{\mu} (\eps_k > -\infty)$, where $\mathcal{M}^{ARUM-E}(U)$ denotes distributions that rationalize $p$ according to ARUM-E when $p$ is identified over $U$. Define $\Theta^{k,CS}(U)$ analogously.
\begin{cor}[Discontinuous Identification] \label{cor:disc}
Let $\{ U^s \}_{s = 1}^{\infty}$ be an increasing collection of subsets of $\mathbb{R}^K$, i.e. $U^s \subseteq U^m$ for $s \leq m$. Assume further that each set $U^s$ is a compact rectangle, i.e. $U^s = \times_{j = 1}^K [\underline{u}^s_j, \overline{u}_j^s]$. Let $U^{\infty} := \cup_{s = 1}^{\infty} U^s$.
\begin{enumerate}[(i)]
    \item For any $k$,
    \[
    \left[\sup_{u \in U^{\infty}} p_k(u),1 \right] \subseteq \cap_{s = 1}^{\infty} \Theta^{k,E}(U^s) = \cap_{s = 1}^{\infty} \Theta^{k,CS}(U^s).
    \]
    \item For any $k$ that can be made extremely attractive with domain $U^{\infty}$,
    \[
    \sup_{u \in U^{\infty}} p_k(u) = \Theta^{k,E}(U^{\infty}) = \Theta^{k,CS}(U^{\infty}).
    \]
\end{enumerate}

\end{cor}
Putting these results together with Proposition~\ref{prop:id}(iii), we conclude that if $U^{\infty}$ is rich enough so that $k$ can be made extremely attractive, the left hand side of Corollary~\ref{cor:disc}(i) is a point if and only if $k$ is considered with probability $1$. In general, there is a severe discontinuity of identification contrasting the case of bounded variation in utility indices with the limit of full variation. For any set $U^s$, no matter how big, the minimal width of the identified set is $1 - \sup_{u \in U^{\infty}} p_k(u)$, but in the limit it is a single point. Thus, it is not the case that large bounded variation resembles the case of unbounded variation. See \cite{magnac2007identification} and \cite{khan2010irregular} for other contexts with such a discontinuity, and how this makes estimation challenging.

The rest of the paper keeps $U$ fixed.

\begin{cor}[Distinguishability] \label{ref:cor1}
Assume $p$ is consistent with ARUM-E or equivalently ARUM-CS. Assume each alternative can be made extremely attractive, i.e. $\mathcal{K}^{EA} = \{1, \ldots, K\}$. The following are equivalent:
\begin{enumerate}[(i)]
    \item $p$ is consistent with ARUM.
    \item $\sup_{u \in U} p_k(u) = 1$ for each $k$.
\end{enumerate}
\end{cor}
This characterizes the extra empirical restrictions of ARUM.

\begin{cor}[Nontrivial Rationalizations] \label{ref:cornontrivial}
Assume $p$ is consistent with either ARUM-E or ARUM-CS, and that for some $k$, $\sup_{u \in U} p_k(u) < 1$. Assume $U$ has bounded utility differences and contains a $k$-maximal point. It follows that:
\begin{enumerate}[(i)]
    \item $p$ can be rationalized according to ARUM-E with a distribution $\mu \in \mathcal{M}^{ARUM-E}$ that satisfies
    \[
    \Pr_{\mu}(\eps_k > -\infty) < 1.
    \]    
    \item $p$ can be rationalized according to ARUM-CS with a distribution $\nu \in \mathcal{M}^{ARUM-CS}$ that satisfies
    \[
    \Pr_{\nu}(k \in S(\eta)) < 1.
    \]    
\end{enumerate}
\end{cor}

This establishes that if the choice probability for some alternative $k$ is separated from $1$, then we we can rationalize $p$ with a distribution in which some alternative is not considered with positive probability. Thus, it is ``as if'' some individuals do not consider all alternatives.

\section{Counterfactuals} \label{sec:counter}
This section builds on the previous results to study implications for counterfactuals. I analyze two types of counterfactual interventions. The first is a utility index intervention, which involves setting the utility index $u$ to a new value, holding the distribution of unobservables fixed. If the utility index for each alternative is its (negative) price, this amounts to a counterfactual price change. The second is an attention intervention, which changes consideration set formation, holding everything else fixed. When we identify probabilities over a region with bounded utility differences, the models all deliver the same counterfactual bounds for utility-interventions, but deliver contrasting bounds for attention interventions.

\subsection{Utility Index Interventions}
Suppose the analyst knows $p$ over $U$, and is interested in theory-consistent counterfactuals at a new value $u^C \not\in U$. I formulate counterfactual restrictions in terms of an extension $p^C : U \cup \{u^c\} \rightarrow \Delta^K$ that agrees with $p$ on the set $U$. That is, $p^C(u) = p(u)$ for $u \in U$. The interpretation is $p^C(u^C)$ represents choice probabilities at the new value of utility indices. If $p^C$ is consistent with ARUM, say it is ARUM-consistent, and similarly for the other models. The set of ARUM-consistent functions is given by $\mathcal{P}_{ARUM}^C$. For ARUM, define the identified set for counterfactuals at the new value $u^C$ as
\[
\Theta^{C}_{ARUM} = \cup_{p^C \in \mathcal{P}^C_{ARUM}} p^C(u^C).
\]
For any element $a \in \Theta^{C}_{ARUM}$, there is a distribution over $\eps$ that matches the known probabilities over $U$, and that yields $a$ as the choice probability at the new value $u^C$. For ARUM-E and ARUM-CS, define the identified sets analogously.
\begin{prop} \label{prop:utilityintervention}
\begin{enumerate}[(i)]
    \item If $U \subseteq \mathbb{R}^K$ has bounded utility differences, then for any $u^C \in \mathbb{R}^K$,
    \[
    \Theta^C_{ARUM} = \Theta^{C}_{ARUM-E} = \Theta^{C}_{ARUM-CS}.   
    \]
    \item For any $U \subseteq \mathbb{R}^K$ and $u^C \in \mathbb{R}^K$,
    \[
    \Theta^{C}_{ARUM-E} = \Theta^{C}_{ARUM-CS}.
    \]
\end{enumerate}
\end{prop}

Part (i) states that for utility interventions, the identified sets for counterfactuals are the same for the three models when utility differences are bounded. Note that for general unbounded $U \subseteq \mathbb{R}^K$, it is possible that
\[
\Theta^C_{ARUM} \neq \Theta^{C}_{ARUM-E}
\]
because the left hand side can be empty when the right hand side is not. This can happen when $p$ is consistent with ARUM-E but not ARUM. In this case, there is no extension of $p$ to the counterfactual point $p^C(u^C)$ using the model ARUM, and so $\Theta^C_{ARUM}$ is empty.

\subsection{Attention Interventions} \label{sec:attention}
An attention intervention changes consideration set formation without changing preferences. I record a first fact.
\begin{fact}
Attention interventions in ARUM do not change structural choice probabilities.
\end{fact}
This is no longer true in consideration set models and so I study the scope for such interventions. In ARUM-CS, I model an intervention by changing the mapping $S(\cdot)$, which maps latent attention factors to consideration sets. A $k$-attention intervention makes alternative $k$ always feasible. That is, $S$ changes to $\tilde{S}^k$, which is $S$ with $k$ added to the consideration set, i.e. $\tilde{S}^k(\eta) = S(\eta) \cup \{k\}$. Given an attention intervention, an admissible counterfactual is given by a counterfactual probability function $p^{\tilde{S}^k} : U \rightarrow \Delta^K$. This $p^{\tilde{S}^k}$ must satisfy the property that there is a distribution $\nu \in \mathcal{M}^{ARUM-CS}$ over $(\eps,\eta)$ such that for any $u \in U$ and $j$,
\[
p^{\tilde{S}^k}_j(u) = \Pr_{\nu} \left( \left\{ u_j + \eps_j > \max_{\ell \in \tilde{S}^k (\eta) : \ell \neq k} u_{\ell} + \eps_{\ell} \right\} \right),
\]
and that
\[
p_j(u) = \Pr_{\nu} \left( \left\{ u_j + \eps_j > \max_{\ell \in S(\eta) : \ell \neq k} u_{\ell} + \eps_{\ell} \right\} \cap \{ j \in S(\eta)\} \right).
\]
In words, there must be a distribution of tastes ($\eps$) and latent attention factors $(\eta)$ that generates the counterfactual probability $p^{\tilde{S}^k}$ while matching the primitive structural choice probability $p$. Let $\mathcal{P}^{Attention}$ be the set of such $p^{\tilde{S}^k}$ functions.

I study bounds on structural choices probabilities before and after a $k$-attention intervention. I analyze bounds on the quantity
\[
p^{\tilde{S}^k}_k(u) - p_k(u),
\]
where $p^{\tilde{S}^k}$ is an admissible counterfactual mapping. An obvious feature is that this difference is nonnegative, because adding $k$ to the consideration set cannot make its choice probability go down. I characterize the highest it can be below.
\begin{prop} \label{prop:boundedattention}
Assume $U$ has bounded utility differences and $U$ contains a $k$-maximal point. The identified set for the maximal change given a $k$-attention intervention is given by
\[
\cup_{p^{\tilde{S}^k} \in \mathcal{P}^{Attention}} \sup_{u \in U} \left\{ p^{\tilde{S}^k}_k(u) - p_k(u) \right\} = \left[0, 1 - \sup_{u \in U} p_k(u) \right].
\]
\end{prop}

The left and side is the definition of the identified set. The lower bound $0$ states that we cannot rule out that an attention intervention may have no change in choice probabilies. The upper bound states that the attention intervention may shift choice probability for alternative $k$ up by the amount $1 - \sup_{u \in U} p_k(u)$. Thus, the maximum scope for a $k$-attention intervention is higher when $\sup_{u \in U} p_k(u)$ is lower. Indeed, when this supremum is $1$, there is no scope for $k$-attention interventions to change the probability of choosing alternative $k$, because the data indicate $k$ is always considered.

When $U$ has unbounded utility differences, I show a qualitatively different result for any alternative that can be made extremely attractive.
\begin{prop} \label{prop:unboundedattention}
Assume $U \subseteq \mathbb{R}^K$ and that $k$ can be made extremely attractive, i.e. $k \in \mathcal{K}^{EA}$. For every admissible counterfactual $p^{\tilde{S}^k} \in \mathcal{P}^{Attention}$ we have
\[
\sup_{u \in U} \left\{  p^{\tilde{S}^k}_k(u) - p_k(u) \right\} =  1 - \sup_{u \in U} p_k(u).
\]
\end{prop}

This states that the maximum scope for $k$-attention interventions is uniquely identified for any alternative that can be made extremely attractive. Moreover, as long as $\sup_{u \in U} p_k(u) < 1$, we identify that there must be scope for $k$-attention interventions. This contrasts with the bounded utility difference case of Proposition~\ref{prop:boundedattention}, which states that it is always possible that an attention intervention does not change choice probabilities. Thus, there is a severe discontinuity between the unbounded and bounded cases, as in Corollary~\ref{cor:disc}. 

I omit formal analysis of ARUM-E and mention a key consideration that this paper is agnostic about. In ARUM-E, $\eps_k = -\infty$ may be interpreted as alternative $k$ either not being available, not considered, or arbitrarily unattractive. The specific interpretation determines the scope for the specific intervention. For example, if $k$ is always considered but sometimes not available because it is sold out, then an attention intervention may do nothing. Likewise, if $k$ is always considered but arbitrarily unattractive to some people (e.g. $k$ is a bag of peanuts and some people have a severe allergy), then an attention intervention will do nothing.

\section{Welfare} \label{sec:welfare}
I now study welfare analysis using the different ARUM-family models. Paralleling the counterfactual analysis of Section~\ref{sec:counter}, I analyze two welfare questions. The first is how welfare changes when utility indices change from $u$ to $\tilde{u}$, e.g. $u = -p$ is negative price and some prices change. I show that welfare analysis does not depend on which ARUM family model the analyst uses, at least when $U$ is convex. The second question is how welfare changes given an attention intervention that makes each individual always consider an alternative. I establish the striking result that for ARUM-CS, it is not possible to obtain meaningful bounds on welfare changes for interventions that make an alternative always be considered.

Throughout, I work with an average indirect utility notion for welfare changes. For $\mu \in \mathcal{M}^{ARUM}$, define the average indirect utility $V^{ARUM}_{\mu} : \mathbb{R}^K \rightarrow \mathbb{R}$ as
\[
V_{\mu}^{ARUM}(u) = \E_{\mu} \left[ \max_{k} \{ u_k + \eps_k \} - \max_{k} \{ \eps_k \} \right].
\]
Note this function is defined over all of $\mathbb{R}^K$ not just $U$. This $V$ is sometimes called the social surplus function. Subtraction is standard and ensures that this expectation exists and is finite (\cite{sorensen2022mcfadden}). Define ARUM-E identically, except the distribution $\mu^E$ can allow $\eps_k = -\infty$ for some $k$.\footnote{Recall for ARUM-E, a unique maximizer exists $\mu^E$-a.s., so $\max_k \{\eps_k \}$ is finite $\mu^E$-a.s.} For ARUM-CS, define
\[
V_{\nu}^{ARUM-CS}(u) = \E_{\nu} \left[ \max_{k \in S(\eta)} \{ u_k + \eps_k \} - \max_{k} \{ \eps_k \} \right],
\]
where $\nu$ is a distribution over $(\eps,\eta)$.

\subsection{Utility Index Changes}

The following envelope theorem is very useful.
\begin{lemma} \label{lem:wdz}
Assume $p$ is consistent with ARUM, ARUM-E, and ARUM-CS. For any $u \in U$, $\mu \in \mathcal{M}^{ARUM}$, $\mu^E \in \mathcal{M}^{ARUM-E}$, and $\nu \in \mathcal{M}^{ARUM-CS}$, it follows that:
\[
p(u) = \nabla_u V_{\mu}^{ARUM}(u) = \nabla_u V_{\mu^{E}}^{ARUM-E}(u) = \nabla_u V_{\nu}^{ARUM-CS}(u).
\]
\end{lemma}
Differentiability of each $V$ is established in the proof using the fact that the maximizer is unique with probability $1$. I note that Lemma~\ref{lem:wdz} holds for $U$ discrete or even a singleton $\{ u \}$. The equality for ARUM is well-known under different conditions, e.g. \cite{mcfadden1981econometric}.\footnote{For ARUM, \cite{mcfadden1981econometric} and \cite{shi2018estimating} impose a density assumption that rules out utility ties for any $u \in \mathbb{R}^K$. Here we assume ties occur with probability $0$ only for $u \in \mathcal{U}$. See \cite{sorensen2022mcfadden} for a version allowing utility ties (which formally is more general than ARUM defined here), in which case the value function may no longer be differentiable.} Lemma 1 in \cite{allen2019identification} covers ARUM-E under alternative high-level conditions. The ARUM-CS case does not have clear precedent.

I consider the welfare change in moving from $u$ to $\tilde{u}$. For ARUM this is given by
\[
\Delta^{ARUM}(\tilde{u},u,\mu) = V_{\mu}^{ARUM}(\tilde{u}) - V_{\mu}^{ARUM}(u),
\]
and I similarly define $\Delta^{ARUM-E}$ and $\Delta^{ARUM-CS}$. Recall that as defined in Section~\ref{sec:iddisting}, $\mathcal{M}^{ARUM}$ is the set of distributions that rationalize $p$ according to ARUM, and similarly for $\mathcal{M}^{ARUM-E}$ and $\mathcal{M}^{ARUM-CS}$. The identified set for welfare changes in ARUM is given by
\[
\cup_{\mu \in \mathcal{M}^{ARUM}} \Delta^{ARUM}(\tilde{u},u,\mu),
\]
and defined analogously for ARUM-E and ARUM-CS.

I use the envelope theorem above to integrate probabilities and identify differences in average indirect utility.

\begin{prop} \label{prop:welfarepoint}
Assume $p$ is consistent with ARUM, ARUM-E, and ARUM-CS. Assume $U \subseteq \mathbb{R}^K$. Let $u, \tilde{u} \in U$, and assume that for any $\alpha \in [0,1]$, $\alpha u + (1 - \alpha) \tilde{u} \in U$. It follows that
\begin{align*}
\cup_{\mu \in \mathcal{M}^{ARUM}} \Delta^{ARUM}(\tilde{u},u,\mu) & =
\cup_{\mu \in \mathcal{M}^{ARUM-E}} \Delta^{ARUM-E}(\tilde{u},u,\mu) \\
& = 
\cup_{\nu \in \mathcal{M}^{ARUM-CS}} \Delta^{ARUM-CS}(\tilde{u},u,\nu) \\
& = \int_0^1 p(t \tilde{u} + (1 - t) u) \cdot (\tilde{u} - u) dt.
\end{align*}
\end{prop}

This result provides a constructive point identification result for welfare changes that applies to each model. Point identification for ARUM and ARUM-E is known under slightly different conditions, e.g. results in \cite{small1981applied} and Theorem 4 in \cite{allen2019identification}. The novelty here is that (i) we can also identify welfare changes for ARUM-CS, and (ii) all models deliver the same welfare implications for utility changes.

\subsection{Attention Changes}
I now analyze how welfare changes with an attention intervention like in Section~\ref{sec:attention}. In ARUM, attention increases do not change welfare because individuals pay attention to all alternatives. In ARUM-E, the role of attention interventions depends on the interpretation of the event $\eps_k = -\infty$. It could be due to stock out, limited consideration, or just the alternative being arbitrarily undesirable. The interpretation determines the role of attention interventions. I thus focus on ARUM-CS, where this ambiguity does not arise because $\eps_k$ must be finite and this paper interprets $S(\eta)$ as the set of items that are considered.

Define now
\[
V_{\nu}^{ARUM-CS}(u,S) =  \E_{\nu} \left[ \max_{k \in S(\eta)} \{ u_k + \eps_k \} - \max_{k} \{ \eps_k \} \right].
\]
Here, $S$ is added as an argument because I consider changes that make an alternative available. That is, I consider a $k$-attention intervention that appends $k$ to the consideration set via $\tilde{S}^k(\eta) = S(\eta) \cup \{ k \}$. The welfare change given a $k$-attention intervention is
\[
V_{\nu}^{ARUM-CS}(u,\tilde{S}^k) - V_{\nu}^{ARUM-CS}(u,S).
\]
The identified set for this difference is defined across distributions $\nu$ that rationalize $p$ with the model ARUM-CS. I characterize it as follows.
\begin{prop} \label{prop:attentionwelfare}
Assume $p$ is consistent with ARUM-CS. Let $U \subseteq \mathbb{R}^K$, $u \in U$, and let $k$ be a specific alternative.
\begin{enumerate}[(i)]
    \item If $\sup_{u \in U} p_k(u) = 1$, then the identified set for a $k$-attention intervention is
    \[
    \cup_{\nu \in \mathcal{M}^{ARUM-CS}} \left\{V_{\nu}^{ARUM-CS}(u,\tilde{S}^k) - V_{\nu}^{ARUM-CS}(u,S) \right\} = 0.
    \]
    \item If $\sup_{u \in U} p_k(u) < 1$ and $U$ is bounded and contains a $k$-maximal point, or alternative $U$ is unbounded and $k$ can be made extremely attractive, then the identified set for a $k$-attention intervention is
    \[
    \cup_{\nu \in \mathcal{M}^{ARUM-CS}} \left\{V_{\nu}^{ARUM-CS}(u,\tilde{S}^k) - V_{\nu}^{ARUM-CS}(u,S) \right\} = [0,\infty).
    \]
\end{enumerate}
\end{prop}

Part (i) states that when alternative $k$ takes probability arbitrarily close to $1$, then an attention intervention does not alter welfare. Part (ii) is striking, since it states that the only restriction is that higher attention cannot decrease utility. The precise values of $p$ are irrelevant when $\sup_{u \in U} p_k(u) < 1$. This indicates a fundamental limit of welfare analysis in this setting. This limit arises because when an alternative is not considered, its utility shock may be arbitrarily high. Thus, making it be considered may lead to an arbitrarily high increase in indirect utility.

\section{Unknown Utility Indices} \label{sec:unknownutility}

The previous analysis focused on the case in which utility indices are known. In practice, it is common to consider covariates that vary and alter utility in a way that is not known \textit{a priori}. This section extends a snapshot of the previous results to this setting.

I describe the modified setup. Instead of the utility index $u_k$ for alternative $k$, replace it with $v_k(x_k)$, where $v_k$ is an unknown function and $x_k$ is a vector of observable regressors that alter the desirability of alternative $k$. Such regressors could be different across alternatives (such as alternative characteristics) or common (such as demographic variables). This setup allows $v_k$ to be linear but does not require this. All covariates are collected in $x = (x'_1, \ldots, x'_K)'$. Each $x_k \in \mathbb{R}^{d_k}$ and the overall vector satisfies $x \in \mathbb{R}^d$ where $d = \sum_{k = 1}^K d_k$. Utility indices are collected as $v(x) = (v_1(x_1), \ldots, v_K(x_K))'$. The structural probabilities $p : U \rightarrow \Delta^K$ are not directly identified here. Instead, I consider the mapping $\tilde{p}(x) = p(v_1(x_1), \ldots, v_K(x_K))$. Then $\tilde{p}_k(x)$ is the probability of choosing alternative $k$ given covariates $x$. I assume $\tilde{p} : \mathcal{X} \rightarrow \Delta^K$ is known over $\mathcal{X}$. One can interpret $\mathcal{X}$ as the observable range over which covariates vary.

\subsection{Distinguishability and Identification of Utility Indices}
I study when it is possible to distinguish between models when utility indices are not known in advance, and the related question of identification of $v$ for each model.
\begin{defn}
Choice probabilities $\tilde{p}$ are consistent with ARUM with utility function $v : \mathcal{X} \rightarrow \mathbb{R}^K$ if there exists a function $p : v(\mathcal{X}) \rightarrow \Delta^K$ such that $p$ is consistent with ARUM and $\tilde{p}(x) = p(v(x))$ for each $x \in \mathcal{X}$.
\end{defn}
The set $v(\mathcal{X})$ is the range of $v$. Consistency with ARUM-E and ARUM-CS is defined analogously.
\begin{prop} \label{prop:covsameu}
Assume $\mathcal{X} \subseteq \mathbb{R}^d$ and that $\tilde{p}$ is consistent with either ARUM, ARUM-E, or ARUM-CS with a utility $v$ that is bounded over $\mathcal{X}$. It follows that $\tilde{p}$ is consistent with ARUM, ARUM-E, and ARUM-CS with the same utility $v$.
\end{prop}

An important case is when $\mathcal{X}$ is compact and we restrict attention to continuous utility functions $v$. This means $v$ is bounded over $\mathcal{X}$ and so any such $v$ that rationalizes one model can rationalize any other. I formalize this with some additional notation. To that end, let $\mathcal{U}$ denote a set of candidate utility functions from $\mathcal{X}$ to $\mathbb{R}^K$, which forms the parameter space. For the model ARUM, define
\[
\Theta^{ARUM}_{\mathcal{U}} = \{ v \in \mathcal{U} \mid \tilde{p} \text{ is consistent with ARUM with utility } v \}.
\]
Define this identified set similarly for each model.
\begin{cor} \label{cor:covariate}
Assume $\mathcal{X} \subseteq \mathbb{R}^d$ is compact and $\mathcal{U}$ consists of continuous functions.
\begin{enumerate}[(i)]
    \item (Identified Sets are Equal.) \[
    \Theta^{ARUM}_{\mathcal{U}} = \Theta^{ARUM-E}_{\mathcal{U}} = \Theta^{ARUM-CS}_{\mathcal{U}}.
    \]
    \item (Point Identification.) If one of the sets $\Theta^{ARUM}_{\mathcal{U}}$, $\Theta^{ARUM-E}_{\mathcal{U}}$ and $\Theta^{ARUM-CS}_{\mathcal{U}}$ is a singleton, then all are singletons.
    \item (Observational Equivalence.) Assume $\tilde{p}$ is consistent with one of the models ARUM, ARUM-E, or ARUM-CS for some $v \in \mathcal{U}$. It follows that $\tilde{p}$ is consistent with all of the models ARUM, ARUM-E, and ARUM-CS for the same $v \in \mathcal{U}$.
\end{enumerate}
\end{cor}
Part (i) states the identified sets are the same here. Part (ii) states point identification in one model is equivalent to point identification in all. For example, \cite{allen2019identification} identify utility indices in a generalization of ARUM-E allowing compact support for regressors and requiring continuous utility indices. Part (iii) states the models are observationally equivalent here. This indicates that counterfactual analysis involving covariate varition also coincides.

Without compactness of $\mathcal{X}$, ARUM is no longer observationally equivalent to the other models, and the identified sets can differ. Lack of observational equivalence is true for the simple example $\mathcal{U}$ equal to the identity mapping (so that covariates are scalar for each alternative), and $\mathcal{X} = \mathbb{R}^K$. Then from Corollary~\ref{ref:cor1} we conclude ARUM is not observationally equivalent to ARUM-E or ARUM-CS. To see why identified sets can differ when $\mathcal{X}$ is not compact, consider a setting in which $\sup_{x \in \mathcal{X}} \tilde{p}_k(x) < 1$ for each $k$. ARUM requires $v$ with bounded utility differences, whereas ARUM-E and ARUM-CS do not.

\subsection{Identification of Marginal Consideration Probabilities}
I now turn to identification of marginal consideration probabilities $\Pr(\eps_{k} > -\infty)$ and $\Pr(k \in S(\eta))$ when utility indices are not known in advance.
\begin{prop} \label{prop:marginalcovariate}
Assume $\mathcal{X} \subseteq \mathbb{R}^{d}$ is the Cartesian product of $d$ compact sets, each in $\mathbb{R}$, and that $\mathcal{U}$ consists of continuous functions. Assume $\tilde{p}$ is known over $\mathcal{X}$ and is consistent with ARUM-E for some $v \in \mathcal{U}$. The identified sets for $\Pr(\eps_{k} > -\infty)$ and $\Pr(k \in S(\eta))$ are equal, and are given by
    \[
    \left[\sup_{x \in \mathcal{X}} \tilde{p}_k(x) , 1 \right].
    \]
\end{prop}
This result provides conditions under which identification of utility indices is separate from identification of these marginal consideration set probabilities: sharp bounds depend only on $\sup_{x \in \mathcal{X}} \tilde{p}_k(x)$.

\section{Discussion} \label{sec:disc}

How does limited consideration alter analysis relative to full-consideration models? The answer depends on the baseline model of choice and the extension to allow limited consideration. This paper asks the question with the baseline additive random utility models and two extensions that allow exogenous limited consideration. The main motivation for studying these models is that despite its limitations, the additive random utility model is widely used, and the extensions provide a natural ``update'' to the classic model.

The broad finding of this paper is that with the exception of attention interventions -- which clearly require a model of limited attention -- the classic and updated models deliver \textit{identical} answers to certain questions in empirically relevant settings. The classic nonparametric model can thus be interpreted as a consideration set model. The sharpest contrast is measuring welfare with an attention intervention that makes one alternative be considered. The consideration set model only concludes the welfare contrast is somewhere in the trivial set $[0,\infty)$, while the full attention model states it must be zero. 

I close with further remarks and qualifications.

This paper adds ARUM to the list of consideration set models. The union of the empirical content of models that have been used to study consideration sets is now large. I illustrate by contrasting two other models. For general random utility (with nonseparable unobservables as in \cite{block1960random}), regularity states that if an alternative is added to the choice set, the probability of choosing an existing alternative cannot go \textit{up}. The model studied in \cite{manzini2014stochastic} is a random utility model and requires regularity, while the model of \cite{cattaneo2020random} requires violations of regularity to identify limited consideration. Comparing these two models shows that whether regularity is or is not a hallmark of limited consideration is thus model specific.

Overall, it seems that the ``right'' model is setting-question specific. ARUM is a reasonable candidate if the analyst wishes to do classic analysis that is robust to exogenous consideration sets. An important caveat is that this paper shows this only for the \textit{nonparametric} model studied here. This paper is silent on comparison of a parsimonious ARUM model with a parsimonious consideration set model. In order for ARUM to mimic a consideration set model, it must have fat tails for the distribution of shocks. Loosely, if $\eps_k = -\infty$ with positive probability, ARUM can only mimic this with a distribution of shocks that allows very negative values. Traditional econometric analysis rules this out. For example, common parametric choices like logit and normal shocks have very thin tails, and may be ill-suited when consideration sets may be relevant.

\bibliographystyle{plainnat}
\bibliography{ref}

\begin{appendices}	

\counterwithin{theorem}{section}
\counterwithin{prop}{section}
\counterwithin{lemma}{section}
\counterwithin{cor}{section}
\counterwithin{assm}{section}
\counterwithin{defn}{section}
\counterwithin{remark}{section}

\section{Proof of Main Results}

\subsection{Proofs for Section~\ref{sec:obs}}

\begin{proof}[Proof of Proposition~\ref{prop:obseq}.]
First I prove part (i). Suppose first that $p$ is consistent with ARUM-E with distribution $\mu$ over $\eps$. Define $S(\eps)$ as the subset of $\{1, \ldots, K \}$ for which $\eps_k > -\infty$. Let $\tilde{\eps} = \eps$ except for any $k$ and event $\eps_k = -\infty$, $\tilde{\eps}_k = 0$ instead so that $\tilde{\eps}$ is finite. This $\tilde{\eps}$ modification is done because ARUM-CS requires finite-valued shocks. Let $\nu$ be the distribution over $(\tilde{\eps},\eps)$ induced by $\mu$. By the definition of ARUM-E, for each $\mu \in \mathcal{M}^{ARUM-E}$, some $\eps_k$ is finite $\mu$-a.s. It is immediate that
\[
\Pr_\mu \left(u_k + \eps_k > \max_{j \neq k } u_j + \eps_j \right) = \Pr_{\nu} \left(u_k + \tilde{\eps}_k > \max_{j \in S(\eps) : j \neq k} u_j + \tilde{\eps}_j \right) = p_k(u).
\]
Since these equalities hold for each $k$ and $u \in U$, $p$ is consistent with ARUM-CS. Note further that $\Pr_{\mu}(\eps_k > -\infty) = \Pr_{\nu}(k \in S(\eps))$.

Now instead suppose $p$ is consistent with ARUM-CS with distribution $\pi$ over $(\eps,\eta)$. Construct a new distribution $\rho$ over $\eps$ by taking the marginal distribution of $\eps$ according to $\nu$ and modifying it by setting $\eps_k = -\infty$ if $k \not\in S(\eta)$. We have
\[
\Pr_\pi \left(u_k + \eps_k > \max_{j \neq k} u_j + \eps_j \right) = \Pr_{\rho} \left(u_k + \eps_k > \max_{j \in S(\eta) : j \neq k} u_j + \eps_j \right) = p_k(u).
\]
Thus, $p$ is consistent with ARUM-E. Moreover, note that $\Pr_{\pi}(k \in S(\eta)) = \Pr_{\nu}(\eps_k > -\infty)$.

I now prove part (ii). Note from (i) one need only show that ARUM is observationally equivalent to ARUM-CS. It is clear that ARUM is a special case of ARUM-CS (by setting $\eta$ degenerate and $S(\eta) = \{1, \ldots, K\}$). Assume $p$ is consistent with ARUM-CS. To complete the proof it remains to show $p$ is consistent with ARUM. Let $\tau \in \mathcal{M}^{ARUM-CS}$ rationalize $p$ for ARUM-CS. For each $(\eps,\eta)$ such that $S(\eta)$ is nonempty, construct $\tilde{\eps}$ so that $\tilde{\eps}_k = \eps_k$ if $k \in S(\eta)$. For $k \not\in S(\eta)$, let $\tilde{\eps}_k$ be negative enough so that $k$ is never chosen for any $u$.\footnote{The following choice suffices:
\[
\tilde{\eps}_k =  \underline{u}(\eps,\eta) := \inf_{u \in U} \min_{k \in S(\eta)} \{ u_k - u_1 + \eps_k \} - \sup_{u \in U} \{ u_k - u_1 \} - 1.
\]
We see $\tilde{\eps}_k$ is finite because $U$ has bounded utility differences bounded.}

We obtain that provided $S(\eta)$ is nonempty, for any $k$ we have
\begin{align*}
\left\{ k \in S(\eta) \right\} & \cap \left\{ u_k  + \eps_k > \max_{j \in S(\eta) : j \neq k} \{ u_j + \eps_j \} \right\} \\
& = \left\{ k \in S(\eta) \right\} \cap \left\{ u_k - u_1 + \eps_k > \max_{j \in S(\eta) : j \neq k} \{ u_j - u_1 + \eps_j \} \right\} \\ 
& =  \left\{ k \in S(\eta) \right\} \cap \left\{ u_k - u_1 + \tilde{\eps}_k > \max_{j \neq k} \{ u_j - u_1 + \tilde{\eps}_j \} \right\} \\
& = \left\{ u_k - u_1 + \tilde{\eps}_k > \max_{j \neq k} \{ u_j - u_1 + \tilde{\eps}_j \} \right\} \\
& = \left\{ u_k + \tilde{\eps}_k > \max_{j \neq k} \{ u_j + \tilde{\eps}_j \} \right\}.
\end{align*}
The first equality subtracts a finite $u_1$ from both sides of an inequality. The second equality replaces $\eps$ with $\tilde{\eps}$. Note that when $S(\eta)$ is nonempty, this does not alter the events since the only candidates for a maximizer are elements of $S(\eta)$. For the third equality, the higher line is clearly a subset of the lower; the lower line is a subset of the higher line because when $k$ is a maximizer, $k \in S(\eta)$. The final equality adds back $u_1$.

We conclude that since $S(\eta)$ is nonempty $\tau$-a.s., for arbitrary $k$ and $u \in U$ we have
\begin{align*}
\Pr_{\tau} \bigg(u_k + \tilde{\eps}_k > \max_{j \neq k} \{ u_j + \tilde{\eps}_j \} \bigg)  & = \Pr_{\tau} \bigg( \left\{ k \in S(\eta) \right\} \cap \left\{ u_k + \eps_k > \max_{j \in S(\eta) : j \neq k} \{ u_j + \eps_j \} \right\} \bigg) \\
& = p_k(u).
\end{align*}
The distribution $\tau$ over $\tilde{\eps}$ rationalizes $p$, and hence $p$ is consistent with ARUM. This completes the proof.
\end{proof}

\begin{proof}[Proof of Lemma~\ref{lem:idsetequiv}]
The equality
\[
\Theta^{k,E} = \Theta^{k,CS}
\]
is established in the proof of Proposition~\ref{prop:obseq}(i).

To establish convexity, consider $c, \tilde{c} \in \Theta^{k,E}$. Let $\mu_c$ and $\mu_{\tilde{c}}$ be distributions over $\eps$ that rationalize $p$ according to ARUM-E, and that satisfy $\Pr_{\mu_c} (\eps_k > -\infty) = c$ and similarly for $\tilde{c}$. Let $\alpha \in [0,1]$. Construct a mixture distribution $\mu_{\alpha c + (1 - \alpha) \tilde{c}}$. Here, $\alpha$ is the weight on $\mu_c$ and $(1 - \alpha)$ is the weight on $\mu_{\tilde{c}}$. This mixture distribution satisfies for any $k$ and $u \in U$ that
\begin{align*}
p_k(u) & =\alpha \Pr_{\mu_{c}} \left(u_k + \eps_k > \max_{j \neq k} \{ u_j + \eps_j \}\right) + (1 - \alpha)  \Pr_{\mu_{\tilde{c}}} \left(u_k + \eps_k > \max_{j \neq k} \{ u_j + \eps_j \}\right) \\
    & = \Pr_{\mu_{\alpha c + (1 - \alpha) \tilde{c}}} \left(u_k + \eps_k > \max_{j \neq k} \{ u_j + \eps_j \}\right).
\end{align*}
Thus, the mixture rationalizes $p$ according to ARUM-E. Finally, we see that $\Pr_{\mu_{\alpha c + (1 - \alpha) \tilde{c}}} (\eps_k > -\infty)  = \alpha c + (1 - \alpha) \tilde{c}$. This establishes convexity of $\Theta^{k,E}$ and hence $\Theta^{k,CS}$ because these sets are equal.

To establish the lower bound, let $\nu$ be a distribution over $\eps$ that rationalizes $p$ according to ARUM-E. We have 
\[
p_k(u) = \Pr_{\nu} \left(u_k + \eps_k > \max_{j \neq k} u_j + \eps_j \right) \leq \Pr_{\nu} \left(\eps_k > -\infty \right).
\]
Thus, $\sup_{u \in U} p_k(u) \leq \Pr_{\nu} (\eps_k > -\infty)$.
\end{proof}

\begin{proof}[Proof of Proposition~\ref{prop:id}.]

I first prove part (i). Let $u^1, u^2, \ldots$ be a sequence of vectors in $U$ that satisfies
\[
\min_{j \neq k} \{ u^s_k - u^s_j \} \geq s
\]
for each $s$. Such a sequence exists because $k \in \mathcal{K}^{EA}$. For any $\nu$ that rationalizes $p$ according to ARUM-E, we have
\begin{align*}
\liminf_{s \rightarrow \infty} p_k(u^s) = & \liminf_{s \rightarrow \infty} \Pr_{\nu} \left(\eps_k > \max_{j \neq k} \{ u^s_j - u^s_k + \eps_j \} \right) \\
\geq & \liminf_{s \rightarrow \infty} \Pr_{\nu} \left(\eps_k > \max_{j \neq k} \{ -s + \eps_j \} \right) \\
= & \Pr_{\nu} \left(\cup_{s = 1}^{\infty} \Big\{ \eps_k > \max_{j \neq k} \{ -s + \eps_j \} \Big\} \right) \\
= & \Pr_{\nu} ( \eps_k > -\infty).
\end{align*}
The first line uses the fact that $\nu$ rationalizes $p$. The second line follows because $u^s_j - u^k \leq -s$. The third line uses the fact that the events
\[
B^s = \left\{ \eps_k > \max_{j \neq k} \{ -s + \eps_j \right\}
\]
satisfy $B^s \subseteq B^t$ for $s \leq t$, and so we can use the continuity from below property of probabilities (\cite{billingsley2017probability}, Theorem 2.1(i)). The final equality uses the fact that the events
\[
\cup_{s = 1}^{\infty} \Big\{ \eps_k > \max_{j \neq k} \{ -s + \eps_j \} \Big\}
\]
and
\[
\{ \eps_k > -\infty \}
\]
are the same. 

Lemma~\ref{lem:idsetequiv} yields $\limsup_{s \rightarrow \infty} p_k(u^s) \leq \Pr_{\nu}(\eps_k > -\infty)$, and so $\lim_{s \rightarrow \infty} p_k(u^s) = \Pr_{\nu} (\eps_k > -\infty)$. Again using the lower bound in Lemma~\ref{lem:idsetequiv}, we have $\sup_{u \in U} p_k(u) \leq \Pr_{\nu}(\eps_k > -\infty)$. Because $\sup_{u \in U} p_k(u) \geq \lim_{s \rightarrow \infty} p_k(u^s)$, we conclude $\sup_{u \in U} p_k(u) = \lim_{s \rightarrow \infty} p_k(u^s) $. We thus have $\Theta^{k,E} = \Theta^{k,CS} = \sup_{u \in U} p_k(u)$.

Part (ii) follows from Proposition~\ref{prop:obseq}(ii). Indeed, since $U$ has bounded utility differences and we assume $p$ is consistent with ARUM-E, then it must also be consistent with ARUM. But this implies we cannot refute that $\Pr(\eps_k > -\infty) = 1$ for ARUM-E.

I now prove part (iii) in two steps. Step 1 proves  $\sup_{u \in U} p_k(u) \in \Theta^{k,E}$ and that $\sup_{u \in U} p_k(u)$ is the (greatest) lower bound. Step 2 proves $1 \in \Theta^{k,E}$. Lemma~ \ref{lem:idsetequiv} then yields $\Theta^{k,E} = \Theta^{k,CS} = [\sup_{u \in U} p_k(u), 1]$ by convexity.

\textit{Step 1.} Let $u^* \in U$ be a $k$-maximal point. I work with $\Theta^{k,E}$ for convenience. Consider the event that $k$ is not chosen when utility indices are given by $u^*$, i.e.
\[
A  = \left\{u^*_k + \eps_k < \max_{j \neq k} \{ u^*_j + \eps_j \} \right\},
\]
which is equivalent to
\[
\eps_k < \max_{j \neq k} \{ u^*_j - u^*_k + \eps_j \}.
\]
Since $u^*$ is k-maximal, we obtain that this event implies for any $w \in U$ that
\[
\eps_k < \max_{j \neq k} \{ w_j - w_k + \eps_j \}.
\]
In words, if $k$ is not chosen at $u^*$ (for fixed $\eps$), then it is not chosen for any $w$. Importantly, we can replace $\eps_k$ with any smaller value and the inequalities above are preserved. We now exploit this key feature.

To that end, consider an arbitrary distribution $\mu$ over $\eps$ that rationalizes $p$ according to ARUM-E. Construct $\tilde{\eps}$ that agrees with $\eps$ on $A^c$, i.e. the complement of $A$. For the complement of this event $A$, let $\tilde{\eps}_j = \eps_j$ for $j \neq k$, and $\tilde{\eps}_k = -\infty$. The only way this modification can alter the events characterizing choices is by making alternative $k$ go from being chosen to not chosen. But as discussed above by analyzing events like A, we see that this alteration (of $\eps$ to $\tilde{\eps}$) does not change events in which $k$ is not chosen at $u^*$. Thus, the distribution $\mu$ over $\tilde{\eps}$ rationalizes $p$. 

We have
\begin{align*}
\Pr_{\mu} (\tilde{\eps}_k = -\infty) = & \Pr_{\mu}( u^*_k + \eps_k  < \max_{j \neq k} \{ u^*_j + \eps_j \}) \\
= & 1 - p_k(u^*) + \Pr_{\mu}(u^*_k + \eps_k = \max_{j \neq k} \{ u^*_j + \eps_j \} ) \\
= & 1 - p_k(u^*) \\
= & 1 - \sup_{u \in U} p_k(u).
\end{align*}
The third equality uses the fact that by the definition of ARUM-E consistency, utility ties at the max occur with probability 0. The final equality comes from the assumption that $u^*$ is $k$-maximal, and uses the facts that $p_k$ is weakly increasing in $u_k$, and weakly decreasing in $u_j$ for $j \neq k$. We obtain $\Pr_{\mu}(\tilde{\eps}_k > -\infty) = \sup_{u \in U} p_k(u)$ and hence $\sup_{u \in U} p_k(u) \in \Theta^{E,k}$. Part (i) establishes $\sup_{u \in U} p_k(u)$ provides a lower bound for $\Theta^{E,k}$, and hence we obtain $\sup_{u \in U} p_k(u)$ is in fact the greatest lower bound.
 
\textit{Step 2.} Since $U$ has bounded utility differences, ARUM and ARUM-E are observationally equivalent from Proposition~\ref{prop:obseq}(ii). ARUM corresponds to $\Pr(\eps_k > -\infty) = 1$ for each $k$. Thus, $1 \in \Theta^{k,E}$, completing the proof of part (ii).
\end{proof}

\begin{proof}[Proof of Corollary~\ref{ref:cor1}]
$((i) \implies (ii))$. If $p$ is consistent with ARUM, then there exists a $\mu \in \mathcal{M}^{ARUM-E}$ that rationalizes $p$ according to ARUM-E, and that places probability $1$ on the event $\{ \eps_k > -\infty\}$. Thus, $1 \in \Theta^{k,E}$. Proposition~\ref{prop:id}(i) states $\Theta^{k,E} = \sup_{u \in U} p_k(u)$ and so $\sup_{u \in U} p_k(u) = 1$.

$((ii) \implies (i))$. Proposition~\ref{prop:id}(i) yields $\Theta^{k,E} = 1$. Thus, for any $\mu \in \mathcal{M}^{ARUM-E}$ we have $\Pr_{\mu} (\eps_k > -\infty \text{ for each } k) = 1$. Thus $\mu \in \mathcal{M}^{ARUM}$ and so $p$ is consistent with ARUM.
\end{proof}

\begin{proof}[Proof of Corollary~\ref{ref:cornontrivial}.]
This is immediate from Proposition~\ref{prop:id}(iii).
\end{proof}

\subsection{Proofs for Section~\ref{sec:counter}}

\begin{proof}[Proof of Proposition~\ref{prop:utilityintervention}]
Part (i) follows from observational equivalence as formalized in Proposition~\ref{prop:obseq}(ii), and the fact that the set $U \cup \{ u^c\}$ is bounded.

Similarly, Proposition~\ref{prop:obseq}(i) implies for general $U \cup \{ u^c \}$ that
\[
\Theta^C_{ARUM-E} = \Theta^C_{ARUM-CS}.
\]
\end{proof}

\begin{proof}[Proof of Proposition~\ref{prop:boundedattention}]

Proposition~\ref{prop:id}(iii) states that the identified set for $\Pr(k \in S(\eta))$ is $\Theta^{k,E} = [p_k(u^*), 1]$, where $u^*$ is a $k$-maximal element of $U$. We leverage this to characterize the lower and upper endpoints of $\cup_{p^{\tilde{S}^k} \in \mathcal{P}^{Attention}} \sup_{u \in U} \left\{ p^{\tilde{S}^k}_k(u) - p_k(u) \right\} $.

\textit{Lower Bound:}
Since $\Theta^{k,E}$ has highest element 1 (corresponding to $k \in S(\eta)$ with probability 1), we immediately obtain that $\cup_{p^{\tilde{S}^k} \in \mathcal{P}^{Attention}} \sup_{u \in U} \left\{ p^{\tilde{S}^k}_k(u) - p_k(u) \right\}$ contains 0. That is, if $k \in S(\eta)$ with probability $1$ then a $k$-intervention does not change anything since $S(\eta) = \tilde{S}^k(\eta)$ with probability $1$. We conclude also that since a $k$-attention intervention adds $k$ to the choice set, we must have that the probability of choosing $k$ weakly goes up because we do not allow utility ties. This yields that $\cup_{p^{\tilde{S}^k} \in \mathcal{P}^{Attention}} \sup_{u \in U} \left\{ p^{\tilde{S}^k}_k(u) - p_k(u) \right\}$ has no element less than $0$ and hence $0$ is its greatest lower bound.

\textit{Upper Bound:}
Since $\Theta^{k,E}$ has lowest element $p_k(u^*)$, there exists a distribution over $(\eps,\eta)$ that rationalizes $p$ according to ARUM-CS and that satisfies $\Pr(k \in S(\eta)) = p_k(u^*)$. Modify the distribution of $\eps$ as follows: on the event $k \not\in S(\eta)$, replace $\eps_k$ with $\tilde{\eps}_k$ that satisfies $u_k + \tilde{\eps}_k > u_j + \eps_j$ for each $u \in U$ and $j \neq k$. Since $U$ has bounded utility differences this is possible by setting $\tilde{\eps}_k$ high enough. Let $\tilde{\eps}$ be a vector with $\tilde{\eps}_j = \eps_j$ for $j \neq k$. Then this new distribution $(\eta,\tilde{\eps})$ rationalizes $p$ according to ARUM-CS because it does not affect choice probabilities. We have that with this new distribution, on the event $k \not\in S(\eta)$, alternative $k$ has highest utility and thus will be chosen under a $k$ attention intervention when it was not without the attention intervention. The probability of $k \not\in S(\eta)$ is maximized at $1 - p_k(u^*)$ and so we obtain that the probability of choosing $k$ can go up by that amount. This upper bound $1 - p_k(u^*)$ cannot be improved because it maximizes the probability $k \not\in S(\eta)$, and attention interventions can only modify choices on the event $k \not\in S(\eta)$.

\textit{Convexity:} The argument for the bound for the upper endpoint can be reproduced for any element of $\Theta^{k,E}$. Since $\Theta^{k,E}$ is convex we obtain that $\cup_{p^{\tilde{S}^k} \in \mathcal{P}^{Attention}} \sup_{u \in U} \left\{ p^{\tilde{S}^k}_k(u) - p_k(u) \right\}$ is as well.
\end{proof}

\begin{proof}[Proof of Proposition~\ref{prop:unboundedattention}]

A $k$-attention intervention cannot decrease the probability $k$ is chosen. Moreover, an attention intervention can only alter choices on the event $k \not\in S(\eta)$. From this we deduce that for any $u$ and $p^{\tilde{S}^k} \in \mathcal{P}^{Attention}$,
\begin{equation} \label{eq:attentionubound}
0 \leq p^{\tilde{S}^k}_k(u) - p_k(u) \leq \Pr(k \not\in S(\eta) = 1 - \sup_{u \in U} p_k(u),
\end{equation}
where the final equality comes from Proposition~\ref{prop:id}.

Since $k \in \mathcal{K}^{EA}$, there is a sequence $u^1, u^2, \ldots$ of points in $U$ that satisfies
\[
\min_{j \neq k} \{ u^s_k - u^s_j \} \geq s.
\]
For any $p^{\tilde{S}^k} \in \mathcal{P}^{Attention}$ and associated $\nu \in \mathcal{M}^{ARUM-CS}$, and any $u \in U$:
\begin{align*}
\Pr_{\nu} ( u_k + \eps_k > \max_{\ell \neq k} u_j + \eps_j ) & \leq \Pr_{\nu} ( u_k + \eps_k > \max_{\tilde{S}^k(\eta) : \ell \neq k} u_j + \eps_j ) \\
& = p^{\tilde{S}^k}_k(u).
\end{align*}
We obtain that
\[
\liminf_{s \rightarrow \infty} \Pr_{\nu} ( u^s_k + \eps_k > \max_{\ell \neq k} u^s_j + \eps_j ) = 1
\]
because each $\nu \in \mathcal{M}^{ARUM-CS}$ is a distribution over $(\eps,\eta)$ with $\eps$ finite. We conclude
\[
\lim_{s \rightarrow \infty} p^{\tilde{S}^k}_k(u^s) = 1.
\]
From the proof of Proposition~\ref{prop:id}(i),
\[
\lim_{s \rightarrow \infty} p_k(u^s) = 1 - \Pr(k \not\in S(\eta) 
\]
so that
\[
\lim_{s \rightarrow \infty} p^{\tilde{S}^k}_k(u^s) - p_k(u^s) = 1 - (1 - \Pr(k \not\in S(\eta)) = \Pr(k \not\in S(\eta).
\]
With (\ref{eq:attentionubound}) we conclude
\[
\sup_{u \in U} \left\{  p^{\tilde{S}^k}_k(u) - p_k(u) \right\} =  1 - \sup_{u \in U} p_k(u).
\]
\end{proof}

\subsection{Proofs for Section~\ref{sec:welfare}}

\begin{proof} [Proof of Lemma~\ref{lem:wdz}]

We first prove the result for ARUM-E and then prove the result for the other cases. Overall, the proof uses convex-analytic arguments to connect a unique maximizer to a differentiable value function.

Let $\Delta^{K - 1} = \{ y \in \mathbb{R}^K \mid \sum_{k = 1}^K y_k = 1, y_k \geq 0 \}$ be the probability simplex in $K$ dimensions. Write
\begin{equation} \label{eq:convex}
\sup_{y \in \mathbb{R}^K} \sum_{k = 1}^k y_k u_k - f(y,\eps)
\end{equation}
where
\[
f(y,\eps) = \begin{cases} -\sum_{k = 1}^k y_k \eps_k & \text{ if } y \in \Delta^{K - 1} \\
    \infty & \text{ otherwise }.
    \end{cases}
\]
We define multiplication as $a (-\infty) = -\infty$ for any scalar $a > 0$, and recall $-\infty + a = -\infty$ for $a \neq \infty$. Thus, the term $\sum_{k = 1}^k y_k \eps_k$ is finite only when $y_k > 0$ implies $\eps_k \neq -\infty$. Note here we invoke that $\eps_k < \infty$ to yield $f(y,\eps) > -\infty$ for any $y$. We observe that $f(\cdot, \eps) : \mathbb{R}^{K} \rightarrow \mathbb{R} \cup \{\infty\}$ is finite over a closed convex set, namely the set of $y \in \Delta^{K - 1}$ satisfying $y_k = 0$ for any $k$ with $\eps_k = -\infty$. Moreover, $f(\cdot, \eps)$ is affine in $y$ over the region where it is finite. We conclude $f(\cdot, \eps)$ is lower-semicontinuous in $u$, i.e. $\{ y \in \mathbb{R}^K \mid f(y,\eps) \leq \alpha \}$ is closed for any finite $\alpha$.

Let $\overline{u} \in U$. Define the event in which the utility-maximizing choice is unique,
\[
A(\overline{u}) = \left\{ \text{For some } k, \overline{u}_k + \eps_k > \max_{j \neq k } \overline{u}_j + \eps_j \right\}.
\]
Note that on the event $A(\overline{u})$, $\max_{k} \{ \eps_k \} > -\infty$. This means that $f(y,\eps)$ is finite for some $y$ on the event $A(\overline{u})$, and thus $f(\cdot, \eps)$ is a proper convex function that is lower-semicontinuous. Define the convex conjugate of $f(\cdot,\eps)$ as
\[
f^*(u,\eps) = \sup_{y \in \mathbb{R}^K} \sum_{k = 1} y_k u_k - f(y,\eps).
\]

We have verified the conditions of Theorem 23.5 in \cite{rockafellar2015convex}. The equivalence of (b) and (a*) in that theorem yields
\[
\argmax_{y \in \mathbb{R}^K} \sum_{k = 1}^k y_k \overline{u}_k - f(y,\eps) = \partial_u f^*(\overline{u}, \eps),
\]
where $\partial_u f^*(\overline{u},\eps)$ is the subdifferential of $f^*(\cdot, \eps)$ at $\overline{u}$.\footnote{See pp 214-215 in \cite{rockafellar2015convex} for the definition of subgradients and the subdifferential.} On the event $A(\overline{u})$, the argmax set is a singleton, and takes the form $(0, \ldots, 0, 1, 0, \ldots 0)$ where the $1$ denotes which alternative is chosen. Since there is a unique maximizer, $\partial_u f^*(\overline{u}, \eps)$ is a singleton. From Theorem 25.1 in \cite{rockafellar2015convex} we obtain on the event $A(\overline{u})$ that $f^*(u, \eps)$ is differentiable in $u$ at $\overline{u}$ and satisfies
\[
\argmax_{y \in \mathbb{R}^K} \sum_{k = 1}^k y_k \overline{u}_k - f(y,\eps) = \nabla_u f^*(\overline{u},\eps).
\]
For any $\mu^E \in \mathcal{M}^{ARUM-E}$, $A(\overline{u})$ occurs $\mu^{E}$-almost surely. Taking expectations, we conclude that
\[
p(\overline{u}) = \E_{\mu^E} \left[ \argmax_{y \in \mathbb{R}^K} \sum_{k = 1}^k y_k \overline{u}_k - f(y,\eps) \right] = \E_{\mu^E}[ \nabla_u f^*(\overline{u},\eps) ].
\]
Recall
\[
V_{\mu^E}^{ARUM-E} (\overline{u}) = \E_{\mu^E} \left[ \max_{k} \{ u_k + \eps_k\} - \max_k \{ \eps_k \} \right] = \E_{\mu^E} \left[ f^*(\overline{u},\eps) - \max_k \{ \eps_k \} \right].
\]
We next need to exchange differentiation and integration. To that end, note on the event $A(\overline{u})$ that
\[
\left|f^*(\overline{u},\eps) - \max_k \{ \eps_k \} \right| \leq \max_{k} |\overline{u}_k|, 
\]
so in particular the left hand side is integrable with respect to any $\mu^E \in \mathcal{M}^{ARUM-E}$. Recalling that convex conjugates like $f^*(\cdot, \eps)$ are always convex, we obtain from Proposition 2.2 in \cite{bertsekas1973stochastic} that
\[
\partial V_{\mu^E}^{ARUM-E} (\overline{u}) = \E_{\mu^E} [ \nabla_u f^*(\overline{u},\eps) ].
\]
Recall $\partial$ denotes the subdifferential. Note that while subdifferentials may be multi-valued in general, the right hand side is a vector and so the left hand side is too. Since $V_{\mu^E}^{ARUM-E}$ is convex, Theorem 25.1 in \cite{rockafellar2015convex} yields
\[
\partial V_{\mu^E}^{ARUM-E} (\overline{u}) = \nabla_u V_{\mu^E}^{ARUM-E}(\overline{u}).
\]
We conclude
\[
p(\overline{u}) = \nabla_u V_{\mu^E}^{ARUM-E} (\overline{u}).
\]
This holds for any $\overline{u} \in U$. The ARUM case follows by identical arguments; it is a bit simpler in fact because $\eps$ is finite $\mu$-a.s. for any $\mu \in \mathcal{M}^{ARUM}$.

We next cover ARUM-CS. The proof follows the same steps and so we present a more concise proof. Each $S(\eta)$ defines feasible alternatives in $\{ 1, \ldots, K \}$. Let $\Delta(S(\eta)) \subseteq \mathbb{R}^K$ denote distributions over these alternatives. For example if $S(\eta) = \{1, 2 \}$ then $\Delta(S(\eta))$ consists of probability vectors of the form $(y_1, y_2, 0, \ldots, 0)$ where $y_1 + y_2 = 1$ and each $y_k \geq 0$. We recognize that $\Delta(S(\eta))$ is a closed, convex set.

Let $\overline{u} \in U$. Define the event in which the utility-maximizing choice is unique,
\[
A^{CS}(\overline{u}) = \left\{ \text{For some } k \in S(\eta), \overline{u}_k + \eps_k > \max_{j \in S(\eta), j \neq k} \overline{u}_j + \eps_j \right\}.
\]
For any $\nu \in \mathcal{M}^{ARUM-CS}$, $A^{CS}(\overline{u})$ occurs $\nu$-a.s. and so we condition on this event for the subsequent statements. Now define
\[
\sup_{y \in \mathbb{R}^K} \sum_{k = 1}^K y_k u_k - f(y,\eps,\eta),
\]
where
\[
f(y,\eps,\eta) = \begin{cases} -\sum_{k = 1}^k y_k \eps_k & \text{ if } y \in \Delta(S(\eta)) \\
    \infty & \text{ otherwise }.
    \end{cases}
\]
Again using Theorem 25.1 in \cite{rockafellar2015convex},
\[
\argmax_{y \in \mathbb{R}^K} \sum_{k = 1}^k y_k \overline{u}_k - f(y,\eps,\eta) = \nabla_u f^*(\overline{u},\eps,\eta).
\]
Since the maximizer is unique $\nu$-a.s. we have
\[
p(\overline{u}) = \E_{\nu}[\nabla_u f^*(\overline{u},\eps,\eta)].
\]
From Proposition 2.2 in \cite{bertsekas1973stochastic} we obtain
\[
p(\overline{u}) = \nabla_u V^{ARUM-CS}_{\nu} (\overline{u}) = \E_{\nu}[\nabla_u f^*(\overline{u},\eps,\eta)].
\]
\end{proof}

\begin{proof} [Proof of Proposition~\ref{prop:welfarepoint}]
Let $\mu \in \mathcal{M}^{ARUM}$ and let
\[
h(t) = V^{ARUM}_{\mu}(t \tilde{u} + (1 - t) u),
\]
which is defined for $t \in [0,1]$. From Lemma~\ref{lem:wdz} we see $h$ is differentiable in $t$ because each $t \tilde{u} + (1 - t) u \in \mathcal{U}$ by assumption. Moreover, $h'(t) = p(t \tilde{u} + (1 - t) u) \cdot (\tilde{u} - u)$. By integrating this derivative we obtain
\begin{align*}
\Delta^{ARUM}(\tilde{u},u,\mu) & = V^{ARUM}_{\mu}(\tilde{u}) - V^{ARUM}_{\mu}(u) \\
& = h(1) - h(0) \\
& = \int_0^1 p(t \tilde{u} + (1 - t) u) \cdot (\tilde{u} - u) dt.
\end{align*}
The same argument holds for the other models, again using Lemma~\ref{lem:wdz}.
\end{proof}

\begin{proof} [Proof of Proposition~\ref{prop:attentionwelfare}]
First I prove (i). Since $\sup_{u \in U} p_k(u) = 1$ we have from Lemma~\ref{lem:idsetequiv} that for any $\nu \in \mathcal{M}^{ARUM-CS}$, $\Pr_{\nu} ( k \in S(\eta)) = 1$. Thus, the following equality holds $\nu$-a.s.
\[
\max_{j \in S(\eta)} \{ u_j + \eps_j \} - \max_{j} \{ \eps_j \} = \max_{j \in \tilde{S}^k(\eta)} \{ u_j + \eps_j \} - \max_{j} \{ \eps_j \}
\]
and so by taking expectations,
\[
V_{\nu}^{ARUM-CS}(u,\tilde{S}^k) = V_{\nu}^{ARUM-CS}(u,S).
\]
Since this holds for any $\nu \in \mathcal{M}^{ARUM-CS}$, this proves (i).

We now prove part (ii), beginning with the bounded $U$ case. The lower bound $0$ follows from Proposition~\ref{prop:id}(iii) using the arguments above, since $1$ is in the identified set for $\Pr_{\nu} (k \in S(\eta))$. I next show the upper bound and then complete the proof by invoking convexity.

From Corollary~\ref{ref:cornontrivial}, there exists a $\nu \in \mathcal{M}^{ARUM-CS}$ that satisfies $\Pr_{\nu} (k \in S(\eta)) = 1 - \gamma < 1$. Write the event
\[
A = \{k \not\in S(\eta) \}.
\]
On the event $A$, the value $\eps_k$ does not matter for choice. Thus for any constant $c$ and $\delta > 0$, there is a $\nu^{c,\delta} \in \mathcal{M}^{ARUM-CS}$ that has
\begin{align*}
\Pr_{\nu^{c,\delta}}  \left(\max_{j \in \tilde{S}^k(\eta)} \{ u_j + \eps_j \} - \max_{j \in S(\eta)} \{ u_j + \eps_j \} \geq c \mid A \right) & \geq 1 - \delta \\
\Pr_{\nu^{c,\delta}} (A) = \Pr_{\nu}(A) & = \gamma.
\end{align*}
Indeed, one can take $\nu^{c,\delta}$ to agree with $\nu$ on the complement of $A$, and then to replace $\eps_k$ with $\eps_k + m$ on the event $A$, where $m$ is sufficiently high. Here, $\nu^{c,\delta}$ and $\nu$ agree on the marginal of the consideration variable $\eta$. Thus, $\nu^{c,\delta}$ translates the marginal of $\eps_k$ up by $m$ on the event $A$, and otherwise agrees with $\nu$.
Write
\begin{align*}
\max_{j \in \tilde{S}^k(\eta)} & \{ u_j + \eps_j \} - \max_{j} \{ \eps_j \} = \\
& \max_{j \in S(\eta)} \{ u_j + \eps_j \} - \max_{j} \{ \eps_j \} + \left(\max_{j \in \tilde{S}^k(\eta)} \{ u_j + \eps_j \} - \max_{j \in S(\eta)} \{ u_j + \eps_j \} \right).
\end{align*}
Recall the left hand size is integrable for any measure in $\mathcal{M}^{ARUM-CS}$. Since the term in parentheses is nonnegative, we have 
\[
V_{\nu^{c,\delta}}^{ARUM-CS}(u,\tilde{S}^k) \geq V_{\nu^{c,\delta}}^{ARUM-CS}(u,S) + c(1-\delta)\gamma
\]
whenever $c(1-\delta)\gamma \geq 0$. Recall $\gamma > 0$. Thus, by setting $c(1-\delta)$ sufficiently high we obtain the left hand side can be made arbitrarily high relative to the right hand side. This yields that
\[
\sup_{\nu \in \mathcal{M}^{ARUM-CS}} \left\{ V^{ARUM-CS}_{\nu} (u, \tilde{S}^k) - V^{ARUM-CS}_{\nu}(u,S) \right\} = \infty.
\]
To complete the proof we leverage convexity of the identified set. Indeed, if $\nu \in \mathcal{M}^{ARUM-CS}$ and $\tilde{\nu} \in \mathcal{M}^{ARUM-CS}$, then any mixture also satisfies $\tau_{\alpha \nu + (1 - \alpha) \tilde{\nu}} \in \mathcal{M}^{ARUM-CS}$, where $\alpha \in [0,1]$.\footnote{ See for example the proof of Lemma~\ref{lem:idsetequiv}.} Here $\alpha$ is the weight on $\nu$ and $(1 - \alpha)$ is the weight on $\tilde{\nu}$. We see
\begin{align*}
V^{ARUM-CS}_{\tau_{\alpha \nu + (1 - \alpha) \tilde{\nu}}} (u, \tilde{S}^k) & = \alpha V^{ARUM-CS}_{\nu} (u, \tilde{S}^k) + (1 - \alpha) V^{ARUM-CS}_{\tilde{\nu}} (u, \tilde{S}^k) \\
V^{ARUM-CS}_{\tau_{\alpha \nu + (1 - \alpha) \tilde{\nu}}} (u, S) & = \alpha V^{ARUM-CS}_{\nu} (u, S) + (1 - \alpha) V^{ARUM-CS}_{\tilde{\nu}} (u, S). 
\end{align*}
By taking differences we see that the identified set in the statement of the proposition is convex. This proves (ii) for the bounded $U$ case.

When $U$ is unbounded and $k$ can be made extremely attractive, we can invoke Proposition~\ref{prop:id}(i) to conclude that for any $\nu \in \mathcal{M}^{ARUM-CS}$, we must have $\Pr_{\nu} (k \in S(\eta)) = \sup_{u \in U} p_k(u) < 1$. Repeating the logic of the previous step, the shock $\eps_k$ is unrestricted on the event $\{ k \not\in S(\eta) \}$, and can be arbitrarily negative or positive. This means the change in average indirect utility can only be bounded by $[0,\infty)$.

Finally, note that in fact (ii) holds under the weaker condition that there is some $\nu \in \mathcal{M}^{ARUM-CS}$ such that $\Pr_{\nu} (k \in S(\eta)) < 1$.
\end{proof}

\subsection{Proofs for Section~\ref{sec:unknownutility}}

\begin{proof}[Proof of Proposition~\ref{prop:covsameu}]
Let $\tilde{p}(x) = p(v(x))$, where $p$ is consistent with ARUM, ARUM-E, or ARUM-CS. The set $v(\mathcal{X})$ is bounded, and thus satisfies bounded utility differences. Then from Proposition~\ref{prop:obseq}(ii), $p$ is consistent with ARUM, ARUM-E, and ARUM-CS. Thus, $\tilde{p}$ is consistent with ARUM, ARUM-E, and ARUM-CS with utility $v$.
\end{proof}

\begin{proof}[Proof of Proposition~\ref{prop:marginalcovariate}]
For each $v \in \mathcal{U}$, the set $v(\mathcal{X})$ is a Cartesian product of closed compact sets because $v$ is continuous and $\mathcal{X}$ is a Cartesian product. In particular, for each $k$ there exists a value $x^* \in \mathcal{X}$ such that for each $\tilde{x} \in \mathcal{X}$ and alternative $j$,
\[
v_k(x^*) - v_j(x^*) \geq v_k(\tilde{x}) - v_j(\tilde{x}).
\]
This makes $x^*$ $k$-maximal and so Proposition~\ref{prop:id}(iii) yields the identifed sets as
\[
\left[\tilde{p}_k(x^*) , 1 \right] =    \left[\sup_{x \in \mathcal{X}} \tilde{p}_k(x) , 1 \right].
\]
\end{proof}

\section{Proof of Supplemental Results}

We extend the earlier notation that one alternative can be made extremely attractive. To that end, let $B \subseteq\{1, \ldots, K\}$ be a subset of the set of alternatives. Say that $B$ can be made simultaneously extremely attractive if alternatives in $B$ can simultaneously be made arbitrarily attractive relative to alternatives in the complement $B^c$, i.e. if
\[
\sup_{u \in U} \min_{k \in B} \max_{j \in B^c} \{u_k - u_j \} = \sup_{u \in U} \left\{\min_{k \in B} u_k - \max_{j \in B^c} u_j \right\} = \infty.
\]

Using such sets, we obtain the following implications.
\begin{prop} \label{appprop:arum}
Assume $p$ is consistent with ARUM-E or equivalently ARUM-CS. Consider the following conditions.
\begin{enumerate}[(i)]
    \item $p$ is consistent with ARUM.
    \item For each $B \subseteq\{1, \ldots, K\}$ and for any sequence $(u^m)_{m = 1}^{\infty}$ of points in $U$ satisfying
    \[
    \lim_{m \rightarrow \infty} \left\{\min_{k \in B} u^m_k - \max_{j \in B^c} u^m_j \right\} \rightarrow \infty
    \]
    it follows that
    \[
    \lim_{m \rightarrow \infty} \sum_{k \in B} p_k(u^m) = 1.
    \]
    \item For each $B \subseteq\{1, \ldots, K\}$ that can be made simultaneously extremely attractive,
    \[
    \sup_{u \in U} \sum_{k \in B} p_k(u) = 1.
    \]
    \item For each $B \subseteq\{1, \ldots, K\}$ that can be made simultaneously extremely attractive,
    \[
    \cup_{\mu \in \mathcal{M}^{ARUM-E}} \Pr_{\mu}(\eps_k > -\infty \text{ for some } k \in B) = 1
    \]
    and 
    \[
    \cup_{\nu \in \mathcal{M}^{ARUM-CS}} \Pr_{\nu}(B \cap S(\eta) \neq \emptyset) = 1.
    \]
\end{enumerate}
We have $(i) \implies (ii) \implies (iii) \implies (iv)$. If in addition $\mathcal{K}^{EA} = \{1, \ldots, K \}$ so that any alternative can be made extremely attractive, we have that conditions $(i)-(iv)$ are equivalent.
\end{prop}

\begin{proof}
(Proof that $(i) \implies (ii)$.) Assume $p$ is consistent with ARUM. Let $B \subseteq\{1, \ldots, K\}$ and let $(u^m)_{m = 1}^{\infty}$ be a sequence as in part (ii).

For any distribution $\mu$ that rationalizes $p$ according to ARUM, we have
\begin{align*}
\liminf_{m \rightarrow \infty} \sum_{k \in B} p_k(u^m) & =  \liminf_{m \rightarrow \infty} \Pr_{\mu} \left(\max_{k \in B} u^m_k + \eps_k > \max_{j \in B^c} u^m_j + \eps_j \right) \\
& \geq \liminf_{m \rightarrow \infty} \Pr_{\mu} \left(\max_{k \in B} \left\{ \min_{\ell \in B} \{ u^m_{\ell} \}+ \eps_k \right\} > \max_{j \in B^c} u^m_j + \eps_j \right) \\
& = \liminf_{m \rightarrow \infty} \Pr_{\mu} \left(\max_{k \in B} \left\{ \eps_k \right\} > \max_{j \in B^c} \left\{ u^m_j - \min_{\ell \in B} \{ u^m_{\ell} \} + \eps_j \right\} \right) \\
& \geq \liminf_{m \rightarrow \infty} \Pr_{\mu} \left(\max_{k \in B} \left\{ \eps_k \right\} > \max_{j \in B^c} \left\{ \sup_{s \geq m} \left\{ u^{s}_j - \min_{\ell \in B} \{ u^{s}_{\ell} \} \right\} + \eps_j \right\} \right) \\
& = \Pr_{\mu} \left(\max_{k \in B} \left\{ \eps_k \right\} > -\infty \right) \\ 
& = 1.
\end{align*}
The first line (equality) uses the fact that the events in which $k$ and $j$ are chosen are mutually exclusive. The second line (inequality) replaces $u^m_k$ with a lower value. The third line (equality) moves the finite $\min_{\ell \in B}\{ u^m_{\ell} \}$ to the right hand side. The fourth line (inequality) replaces certain $m$ terms with a supremum over $s \geq m$, which lowers the probability. The fifth line (equality) uses the fact that the events
\[
B^m = \left\{ \max_{k \in B} \{ \eps_k \} > \max_{j \in B^c} \left\{ \sup_{s \geq m} \left\{ u^{s}_j - \min_{\ell \in B} \{ u^{s}_{\ell} \} \right\} + \eps_j\} \right\} \right\}
\]
satisfy $B^m \subseteq B^q$ for $q \geq m$. Since $(u^m)_{m = 1}^{\infty}$ is a sequence as in part (ii) we obtain 
\[
\lim_{m \rightarrow \infty} \left\{ u^{m}_j - \min_{\ell \in B} \{ u^{m}_{\ell} \} \right\} = -\infty.
\]
and hence $\cup_{m = 1}^{\infty} B^m = \{ \max_{k \in B} \{ \eps_k \} > -\infty \}$. The fifth line then follows from the continuity below property of probabilities (e.g. (\cite{billingsley2017probability}, Theorem 2.1(i))). The final equality uses the fact that each $\eps_k$ is finite $\mu$-almost surely (recall $\mu \in \mathcal{M}^{ARUM}$). This proves that (i) implies (ii).

(Proof that $(ii) \implies (iii)$.) Let $(u^m)_{m = 1}^{\infty}$ be a sequence as in (ii). Since $u^m \in U$ for each $m$,
\[
1 = \lim_{m \rightarrow \infty} \sum_{k \in B} p_k (u^m) \leq \sup_{u \in U} \sum_{k \in B} p_k(u) \leq 1.
\]
Thus the inequalities are equalities.

(Proof that $(iii) \implies (iv)$.) Let $B \subseteq \{1, \ldots, K\}$ be a set that can be made simultaneously extremely attractive. Because $p$ is consistent with ARUM-E, there exists some distribution $\mu \in \mathcal{M}^{ARUM-E}$. We have
\begin{align*}
1 = \sup_{u \in U} \sum_{k \in B} p_k(u)
& = \sup_{u \in U} \Pr\left( \max_{k \in B} u_k + \eps_k > \max_{j \in B^c} u^m_j + \eps_j \right) \\
& \leq P_{\mu} \left( \eps_k > -\infty \text{ for some } k \in B \right).
\end{align*}
Similarly, if $\nu \in \mathcal{M}^{ARUM-CS}$, then
\begin{align*}
1 \leq \sup_{u \in U} \sum_{k \in B} p_k(u)
& \leq P_{\nu} \left( B \cap S(\eta) \neq \emptyset \right) = 1.
\end{align*}
This proves that (iii) implies (iv).

Finally, note that when $\mathcal{K}^{EA} = \{1, \ldots, K \}$ and (iv) holds, we can set $B = \{k\}$, which yields that $Pr_{\mu} (\eps_k > -\infty) = 1$ for any $\mu \in \mathcal{M}^{ARUM-E}$. Since $k$ is arbitrary, we obtain that for any $\mu \in \mathcal{M}^{ARUM-E}$, we also have $\mu \in \mathcal{M}^{ARUM}$ and hence $p$ is consistent with ARUM. Thus, (iv) implies (i) and so conditions (i)-(iv) are equivalent.

\end{proof}

\end{appendices}

\end{document}